\documentclass[12pt]{amsart}

\usepackage{amsmath,amssymb,amsbsy}
\usepackage[all]{xy}
\usepackage[pdftex]{graphicx}

\makeatletter
\@addtoreset{equation}{section}
\makeatother

\usepackage{enumerate}

\usepackage{color}

\def\CC{{\mathbb C}}

\def\RR{{\mathbb R}}

\def\PP{{\mathbb P}}

\def\fv{{\mathfrak v}}

\def\tX{{\widetilde X}}

\def\tw{{\widetilde w}}

\def\nuI#1{{\nu^{\mathrm{I}}_{#1}}}

\def\trp{{\, {}^t\!}}

\def\book#1{\rm{#1}, }
\def\paper#1{\textit{#1}, }
\def\jour#1{\rm{#1}, }
\def\yr#1{({\rm{#1}) }}
\def\vol#1{\textbf{#1}}
\def\pages#1{\rm{#1}}

\def\publaddr#1{\rm{#1}, }
\def\publ#1{\rm{#1}, }
\def\by#1{{\rm{#1}, }}

\newtheorem{theorem}{Theorem}[section]

\def\book#1{\rm{#1}, }
\def\paper#1{\textit{#1}, }
\def\jour#1{\rm{#1}, }
\def\yr#1{({\rm{#1}) }}
\def\vol#1{\textbf{#1}}
\def\pages#1{\rm{#1}}

\def\publaddr#1{\rm{#1}, }
\def\publ#1{\rm{#1}, }
\def\by#1{{\rm{#1}, }}

\begin{document}

\title{A graphical representation of hyperelliptic KdV solutions}

\author{Shigeki MATSUTANI}

\date \today
\maketitle


\begin{abstract}
The periodic and quasi-periodic solutions of the integrable system have been studied for four decades based on the Riemann theta functions.
However, there is a fundamental difficulty in representing the solutions graphically because the Riemann theta function requires several transcendental parameters.
This paper presents a novel method for the graphical representation of such solutions from the algebraic treatment of the periodic and quasi-periodic solutions of the Baker-Weierstrass hyperelliptic $\wp$ functions. 
We demonstrate the graphical representation of the hyperelliptic $\wp$ functions of genus two.
\end{abstract}

{\bf{Keyword:}}
hyperelliptic $\wp$ functions;
graphical representation;
periodic and quasi-periodic solutions;
algebro-geometrical method;
Abelian integral;

{\bf{MSC2020}}
14H45; 
14H40; 
14H70; 
35Q53; 



\section{Introduction}
\label{2sec:Heq}

It has been known since the 1970s that periodic and quasi-periodic solutions of the nonlinear integrable equations are described by the Riemann theta functions \cite{BBEIM} as meromorphic functions on the Jacobi varieties associated with algebraic curves of higher genus, known as algebro-geometric solutions.
Due to the properties of the theta functions, many problems related to the nonlinear integrable equations are solved, and such investigations have currently proceeded as a mainstream of the study of the algebro-geometric solutions.
However, while the graphical description of the meromorphic functions on the Riemann sphere of genus zero, called soliton solutions, and elliptic functions of genus one have been possible and have yielded significant results in many applications, graphical representations of genus $g \geq 2$ have been obtained sporadically but are not yet fully understood, e.g., \cite{HI, BBEIM}.

On the other hand, the Riemann theta function defined on the Abelian variety with transcendental parameters given by the values of integrals is determined by a higher dimensional series whose values are complex numbers.
There are fundamental problems in its analytical and numerical treatment, including the dimension of the parameter spaces due to the Schottky problem.
Therefore, it isn't easy to treat the solutions of integrable systems of these higher genera with the same level of concreteness as elliptic functions by means of the approach of the theta function, which is currently the mainstream approach.

To overcome this situation, the author of this paper, together with Emma Previato, has been studying a reformulation of Abelian function theory using the Weierstrass-Baker approach \cite{Baker97}, intending to construct a framework for more concrete treatment of solutions of integrable equations of these higher genera with the same level of concreteness and representability as elliptic functions \cite{MP23}.

In our approach \cite{MP23}, by generalizing the sigma function in Weierstrass elliptic function theory to general algebraic curves, it is found that meromorphic functions on Jacobi varieties can be treated equivalently to the algebraic structure of algebraic curves \cite{KMP22}.
Since the solutions of nonlinear integrable equations are well described by the sigma functions \cite{BEL97b, BEL20, Mat01a, Mat02c}, these studies allowed us to treat hyperelliptic functions graphically in the paper \cite{MP22, Mat23}, though only partially.

In this paper, we demonstrate that this result can be applied to the meromorphic functions on real hyperelliptic curves to graphically represent the hyperelliptic $\wp$ function.
For simplicity, the method presented in this paper is by limiting the genus to two but it can be easily extended to general genus.
The method is extremely simple and clear, which is based on Euler's numerical quadrature method \cite{RLeV}.
Of course, at the existence and derivation of the solutions, we use the findings from the sigma function, but in the final result we do not deal with any theta-series at all, at least, numerically.
Also, the object of the quadrature method is the Abelian integral, which is the integral of the differentials of the first kind.
Since the differentials of the first kind are holomorphic all over the curve, numerical analytic problems are very slight.
Therefore, the parameters of the calculation are extremely simple, and the method fundamentally solves the problems in the conventional approach.

Recently there appear several fascinating studies on the expression of higher degree solutions in terms of the data of elliptic functions, 
Ayano and Buchstaber gave an explicit representation of a certain class of hyperelliptic $\wp$-functions of genus two in terms of double elliptic functions \cite{AyanoBuchstaber22} based on the results of Bolza \cite{Bol87}, and Belokolos and Enolskii \cite{BE2001, BE2002}.
Kakei explicitly shows the so-called elliptic solitons \cite{EEP} using the B\"acklund transformations and demonstrates them graphically using the elliptic $\wp$ functions \cite{Kakei}.
These moves show that we should go beyond the conventional abstract approach to the integral system and move to the concrete treatment of periodic and quasi-periodic solutions of nonlinear integrable equations of higher degrees.
We consider our approach to be one of them.

\bigskip

The contents of this paper are as follows.
Section two shows the relationship between the KdV equation and the hyperelliptic $\wp$ function algebraically.
Section three provides the novel algorithm to obtain the graphical representation of the hyperelliptic $\wp$ function and its demonstrations.
Section four is devoted to the discussion and conclusion.

\section{the KdV equation and the Baker-Weierstrass hyperelliptic $wp$ functions of genus two}

Let us consider a non-degenerate hyperelliptic curve $X$ of genus two ($g=2$),
\begin{gather}
\begin{split}
y^2 &= (x-b_1)(x-b_2) \cdots (x-b_5)\\
    &= x^5 + \lambda_4 x^4 + \cdots + \lambda_1x + \lambda_0\\
\end{split}
\label{eq:hyperCurve1}
\end{gather}
with infinity point $\infty\in X$.
We also introduce the Abelian universal covering $\kappa_X:\tX\to X$ of $X$, and the Abelian differential of the first kind $\nuI{i}=\displaystyle{\frac{x^{i-1}dx}{2y}}$, $(i=1,2)$, $\nuI{}=\trp(\nuI{1}, \nuI{2})$, which is holomorphic on $X$.
We consider the symmetric products $S^2 \tX$ of $\tX$, and $S^2 X$ of $X$, and the Abelian integral $\tw: S^2 \tX \to \CC^2$,
$\displaystyle{\tw(\gamma_1, \gamma_2)=\int_{\gamma_1}
\begin{pmatrix} \nuI1\\ \nuI2 \end{pmatrix}
+
\int_{\gamma_2}
\begin{pmatrix} \nuI1\\ \nuI2 \end{pmatrix}}$.
Let its image be $u ={}^t\! (u_1, u_2)$ such that $\kappa_X\gamma_a=(x_a, y_a)$, $(a=1,2) \in X$.
Thus we also denote it by $\tw((x_1, y_1), (x_2, y_2))$ simply.
Here $\nuI{}$ is the Abelian differentials of the first kind,

As we mentioned for the case of general genus $g$ in \cite{Mat01a, Mat02c}, we show the relation between the KdV equation and the hyperelliptic $\wp$ functions for $X$ of genus two.

By using the Baker-Weierstrass theory of the hyperelliptic functions \cite{Baker97}, we have
\begin{equation}
\wp_{22}(u)=(x_1+x_2) , \quad
\wp_{21}(u)=-x_1x_2,
\label{5eq:JIF_g2}
\end{equation}
where $((x_1, y_1), (x_2, y_2)) \in S^2 X$, $u=\tw((x_1, y_1), (x_2, y_2))$,
$$
\wp_{ij} = -\frac{\partial^2}{\partial u_i u_j} \log \sigma(u),
$$
and $\sigma$ is the hyperelliptic sigma function.
The expression (\ref{5eq:JIF_g2}) is crucial and known as the Jacobi-inversion formulae.
 It means that there is $\zeta_{2}$ such that 
$\displaystyle{\wp_{21}=-\frac{\partial}{\partial u_1}\zeta_2}$ and 
$\displaystyle{\wp_{22}=-\frac{\partial}{\partial u_2}\zeta_2}$.
Though it is a non-trivial fact which should be investigated by the sigma function theory \cite{Baker97,BEL97b, MP23}, by using the fact, it is easy to obtain the following theorem.

Let $\displaystyle{\wp_{ijk\cdots \ell}
:=-\frac{\partial}{\partial u_i}
\frac{\partial}{\partial u_j}
\frac{\partial}{\partial u_k}\cdots
\frac{\partial}{\partial u_\ell}\log \sigma}$.
Then we have the differential identity of genus two as the KdV equation \cite{Mat01a, 02c}:
\begin{theorem}\label{th:KdV2}
\begin{equation}
\wp_{2222}=6 \wp_{22} + 4\wp_{21}+4\lambda_4 \wp_{22}+2 \lambda_3.
\label{5eq:KdV_wpg2}
\end{equation}
By letting $\displaystyle{
\fv(s,t):=-2\wp_{22}\begin{pmatrix}-4 t\\ -4\lambda_4t +s\end{pmatrix}}$, 
$\fv(s,t)$ obeys the KdV equation,
\begin{equation}
\partial_t \fv+ 6\fv \partial_s \fv + \partial_s^3 \fv =0,
\label{5eq:KdV_v_g2}
\end{equation}
where $\displaystyle{\partial_{t}=\frac{\partial}{\partial t}}$ and 
$\displaystyle{\partial_{s}=\frac{\partial}{\partial s}}$.
\end{theorem}

\begin{proof}
Let us prove this algebraically.
Let $u = \tw((x_1, y_1), (x_2, y_2))$.
Due to the definition of the holomorphic one-forms or the differentials of the first kind;
$$
\begin{pmatrix}
d u_1 \\ du_2
\end{pmatrix}
=
\begin{pmatrix}
 1/2y_1 & 1/2y_2\\
 x_1/2y_1 & x_2/2y_2
\end{pmatrix}
\begin{pmatrix}
d x_1 \\ dx_2
\end{pmatrix}.
$$
Let the matrix in the right hand side be denoted by $M$, i.e., $M=(\partial_i u/\partial x_j)$. 
Using its inverse $M^{-1}$, we have
\begin{equation}
\begin{pmatrix}
d x_1 \\ dx_2
\end{pmatrix}
=
\frac{1}{x_2-x_1}
\begin{pmatrix}
 2x_2y_1 & -2y_1\\
 -2x_1y_2 & 2y_2
\end{pmatrix}
\begin{pmatrix}
d u_1 \\ du_2
\end{pmatrix}.
\label{5eq:g2dxdu}
\end{equation}
Further by using $\trp M^{-1}$, we obtain the relation
\begin{equation}
\begin{pmatrix}
\partial_{u_1} \\ 
\partial_{u_2}
\end{pmatrix}
=
\frac{1}{x_2-x_1}
\begin{pmatrix}
 2x_2y_1 & -2x_1y_2\\
 2y_1 & 2y_2
\end{pmatrix}
\begin{pmatrix}
\partial_{x_1} \\ 
\partial_{x_2}
\end{pmatrix},
\label{5eq:g2pupx}
\end{equation}
where $\displaystyle{\partial_{u_a}=\frac{\partial}{\partial u_a}}$, and $\displaystyle{\partial_{x_a}=\frac{\partial}{\partial x_a}}$, $(a=1,2)$.
Hence the differential of $\wp_{22}$ with respect to $u_2$ is interpreted by the algebraic differential of $(x_1+x_2)$,
\begin{gather}
\begin{split}
\frac{\partial}{\partial u_2}(x_1+x_2) &=
\left(
\frac{2y_1}{x_1-x_2}\frac{\partial}{\partial x_1}
+\frac{2y_2}{x_2-x_1}\frac{\partial}{\partial x_2}\right)(x_1+x_2)\\
&=
\frac{2y_1}{x_1-x_2}-\frac{2y_2}{x_1-x_2}.
\end{split}
\label{5eq:g2p2xx}
\end{gather}
Since $y_a^2=f(x_a)$ shows that $2y_a d y_a = f'(x_a) d x_a$, we compute
$\displaystyle{\frac{\partial^2}{\partial u_2^2}(x_1+x_2)}$ and then we obtain
\begin{equation}
\frac{\partial}{\partial x_1}
\left[\frac{2y_1}{x_1-x_2}-\frac{2y_2}{x_1-x_2}\right]=
\frac{-2(y_1+ y_2)}{(x_1-x_2)^2}
+\frac{2f'(x_1)}{2(x_1-x_2)y_1} ,
\end{equation}
\begin{equation}
\begin{split}
\frac{\partial^2}{\partial u_2^2}
(x_1+x_2)
&=
\frac{2y_1}{x_1-x_2}
\left[\frac{-2(y_1+y_2)}{(x_1-x_2)^2}
+\frac{2f'(x_1)}{2(x_1-x_2)y_1}\right]
+(1 \leftrightarrow 2)\\
&=
\frac{-4f(x_1)}{(x_1-x_2)^3}
+\frac{2f'(x_1)}{(x_1-x_2)^2}
+\frac{-4f(x_2)}{(x_1-x_2)^3}
+\frac{2f'(x_2)}{(x_1-x_2)^2}\\
&=
\frac{-2}{(x_1-x_2)^2} I(x_1, x_2),\\
\end{split}
\end{equation}
where
$\displaystyle{
I:=2 \frac{f(x_1)-f(x_2)}{x_1-x_2} - (f'(x_1)+f'(x_2))}$.
Direct computation gives
$$
I=(x_1-x_2)^2[-3(x_1+x_2)^2 + 2 x_1 x_2 - 2\lambda_4 (x_1+x_2)
+\lambda_3].
$$
Then we have
$$
\frac{\partial^2}{\partial u_2^2}
(x_1+x_2)=6(x_1+x_2)^2 +4 x_1 x_2
+4\lambda_4 (x_1+x_2)+2\lambda_3.
$$
 (\ref{5eq:JIF_g2}) shows that it is identified with (\ref{5eq:KdV_wpg2}).
Further since we have
$$
\begin{pmatrix}
\partial_t \\ \partial_s
\end{pmatrix}
=
\begin{pmatrix}
-4 & -4\lambda_4\\ 0 & 1
\end{pmatrix}
\begin{pmatrix}
\partial_{u_1} \\ \partial_{u_2}
\end{pmatrix}
=
\begin{pmatrix}
\partial u_1/\partial t &\partial u_2/\partial t \\ 
\partial u_1/\partial s &\partial u_2/\partial s  
\end{pmatrix}
\begin{pmatrix}
\partial_{u_1} \\ \partial_{u_2}
\end{pmatrix},
$$
$$
\begin{pmatrix}
d u_1 \\ d u_2
\end{pmatrix}
=
\begin{pmatrix}
\partial u_1/\partial t &\partial u_1/\partial s \\ 
\partial u_2/\partial t &\partial u_2/\partial s  
\end{pmatrix}
\begin{pmatrix}
dt \\ ds
\end{pmatrix}
=
\begin{pmatrix}
-4 &0 \\ -4\lambda_4 & 1
\end{pmatrix}
\begin{pmatrix}
dt \\ ds
\end{pmatrix},
$$
and thus for $\fv$, we obtain (\ref{5eq:KdV_v_g2}).
\end{proof}

\section{Graphical representation of $\wp$ function}

We assume that all $b_i$ in (\ref{eq:hyperCurve1}) are real number, $b_i < b_{i+1}$, $(i=1, \ldots, 4)$.
(\ref{5eq:g2dxdu}) shows that for a point $u = \tw((x_1, y_1)) +\tw((x_2, y_2)) \in \CC^2$, $d u_2$ gives the data of $d x_1$ and $d x_2$.

Let us consider $\displaystyle{\begin{pmatrix} 0 \\ d u_2\end{pmatrix}}$ and 
\begin{equation}
\begin{pmatrix}
d x_1 \\ dx_2
\end{pmatrix}
=
\frac{1}{x_2-x_1}
\begin{pmatrix}
-2y_1 d u_2\\
2y_2 d u_2
\end{pmatrix}.
\label{eq:dx_du2}
\end{equation}
If we integrate this differential equation, we get $x_1(u)$ and $x_2(u)$ as a point of $S^2 X$.
Using them, we basically have the solution of the KdV equation, $\wp_{22}(u)=x_1(u) + x_2(u)$.
We compute it by Euler's numerical quadrature \cite{RLeV} as follows.
\begin{enumerate}

\item 
Assume initial conditions $(x_1, x_2)$ such that $x_1 \in (b_1, b_2)$ and $x_1 \in (b_2, b_3)$, and a small parameter $\delta u_2$.

\item Let $u_2 := u_2+ \delta u_2$

\item Compute $\displaystyle{\delta x_1 = \frac{2y_1 \delta u_2}{x_1-x_2}  }$ and $\displaystyle{\delta x_2 = \frac{2y_2 \delta u_2}{x_2-x_1}  }$ from (\ref{eq:dx_du2}).

\item Let $x_1:=x_1+ \delta x_1$ and $x_2:=x_2+ \delta x_2$

\item If $x_1$ goes beyond the branch line $(b_1, b_2)$, go back to the inner direction by changing the leaf of the covering $\varpi_x : X \to \PP$ $(x,y)\mapsto x$, i.e., $y$ to $-y$. Similarly we do so for $x_2$.

\item Go to 2.

\end{enumerate}

\bigskip

We computed several cases for three hyperelliptic curves 
$(b_1, b_2, b_3, b_4, b_5) = (0, 1, 2, 3, 4)$, 
$(0, 1.8, 2, 3, 3.2)$ and $(0, 1.99, 2.0, 3, 3.01)$ as in Figures
\ref{fg:Fig01}, \ref{fg:Fig02} and \ref{fg:Fig03}. 
We put $\delta u_2=1.0 \times 10^{-4}$.
They may show the degenerating family of the curves to the degenerated curves of soliton solution \cite{Mat01b}, i.e., $y^2 = x(x-b_1)^2(x-b_2)^2$ for a certain $b_1, b_2 \in \CC$.
We checked the independence of the numerical parameter $\delta u_2$, which is guaranteed by the properties of the Abelian integrals.

Since in our method, we cannot evaluate the integrals from $\infty \in X$,
we let $u_0:=\tw(B_1+B_4)$, where $B_a:=(b_a,0)$ and 
$$
\tw(B_1, B_4)=\int_\infty^{B_1} \nuI{} + \int_\infty^{B_4} \nuI{}.
$$
The values of the Abelian integrals are obtained as the difference from $u_0$.
Then we obtained $x_1(\begin{pmatrix}0 \\u_2\end{pmatrix}+u_0)$ and 
$x_1(\begin{pmatrix}0 \\u_2\end{pmatrix}+u_0)$ in Figures
\ref{fg:Fig01}, \ref{fg:Fig02} and \ref{fg:Fig03} (a).

Since Theorem \ref{th:KdV2} shows that $\wp_{22}$ is regarded as the algebro-geometric solution of the KdV equation of genus two, we illustrated it in Figures
\ref{fg:Fig01}, \ref{fg:Fig02} and \ref{fg:Fig03}.

Modifying the above algorithm for the $u_1$ direction, we also compute the different $u_1$.
Further $\wp_{22}(\begin{pmatrix}u_1 \\u_2\end{pmatrix}+u_0)$ for $u_1=0.0$,
$1.0$ and $2.0$.
Since the directions of $u_1$ and time $t$ are opposite in the KdV equation, they appear to move from right to left.
The calculations are completed in a few seconds.

\begin{figure}
\begin{center}

\includegraphics[width=0.38\hsize, bb=0 0 497 359]{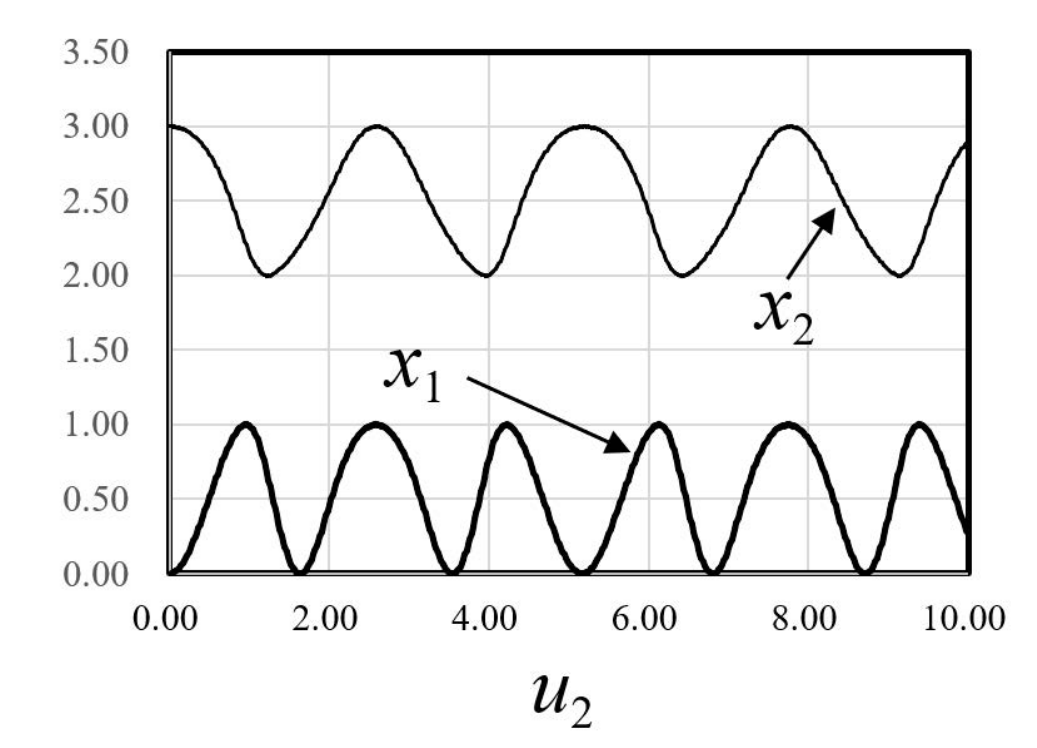}
\hskip 0.1\hsize
\includegraphics[width=0.41\hsize, bb=0 0 596 413]{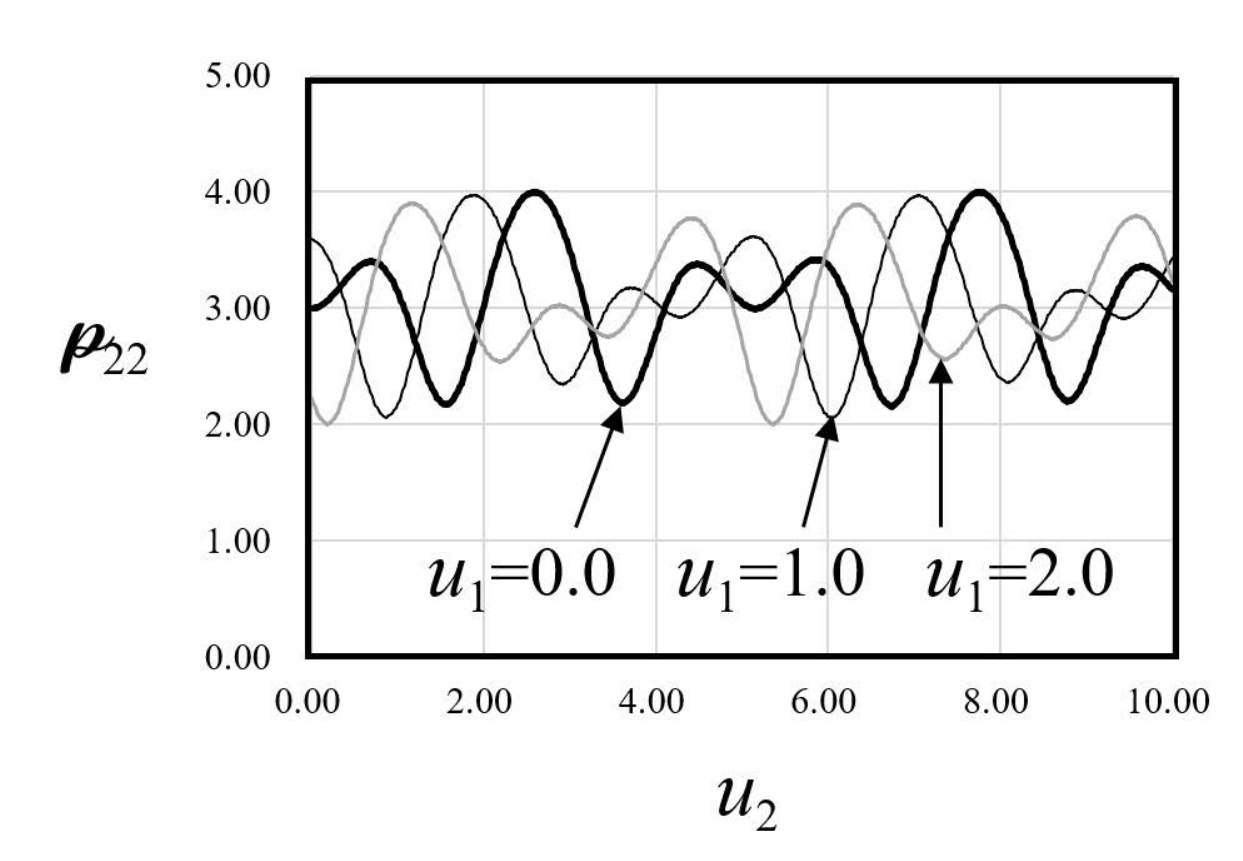}

(a) \hskip 0.45\hsize (b)

\end{center}

\caption{
The graphic representation of hyperelliptic functions of the curve
$(b_1, b_2, b_3, b_4, b_5) = (0, 1, 2, 3, 4)$:
We assume that $u_0:=w(B_1+B_4)$.
(a): $x_1(\begin{pmatrix}0 \\ u_2\end{pmatrix}+u_0)$ and 
$x_1(\begin{pmatrix}0 \\ u_2\end{pmatrix}+u_0)$.
(b): $\wp_{22}(\begin{pmatrix}u_1 \\ u_2\end{pmatrix}+u_0)$ for $u_1=0.0$,
$1.0$, and $2.0$.
}
\label{fg:Fig01}
\end{figure}

\begin{figure}
\begin{center}

\includegraphics[width=0.38\hsize, bb=0 0 503 358]{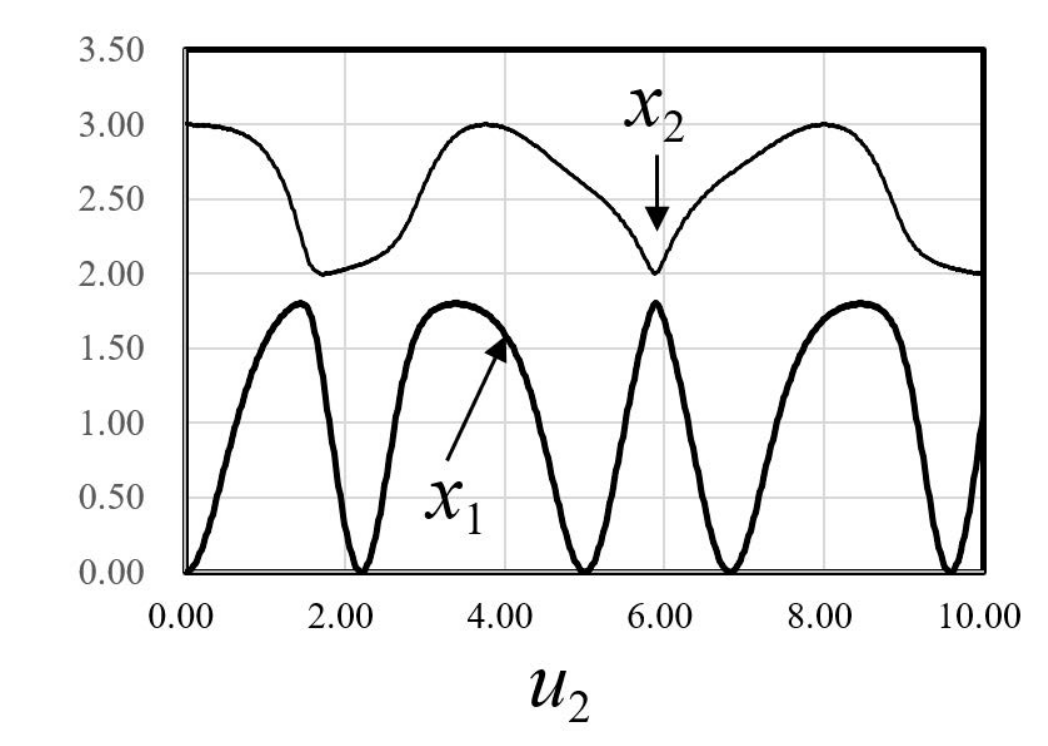}
\hskip 0.1\hsize
\includegraphics[width=0.41\hsize, bb=0 0 574 392]{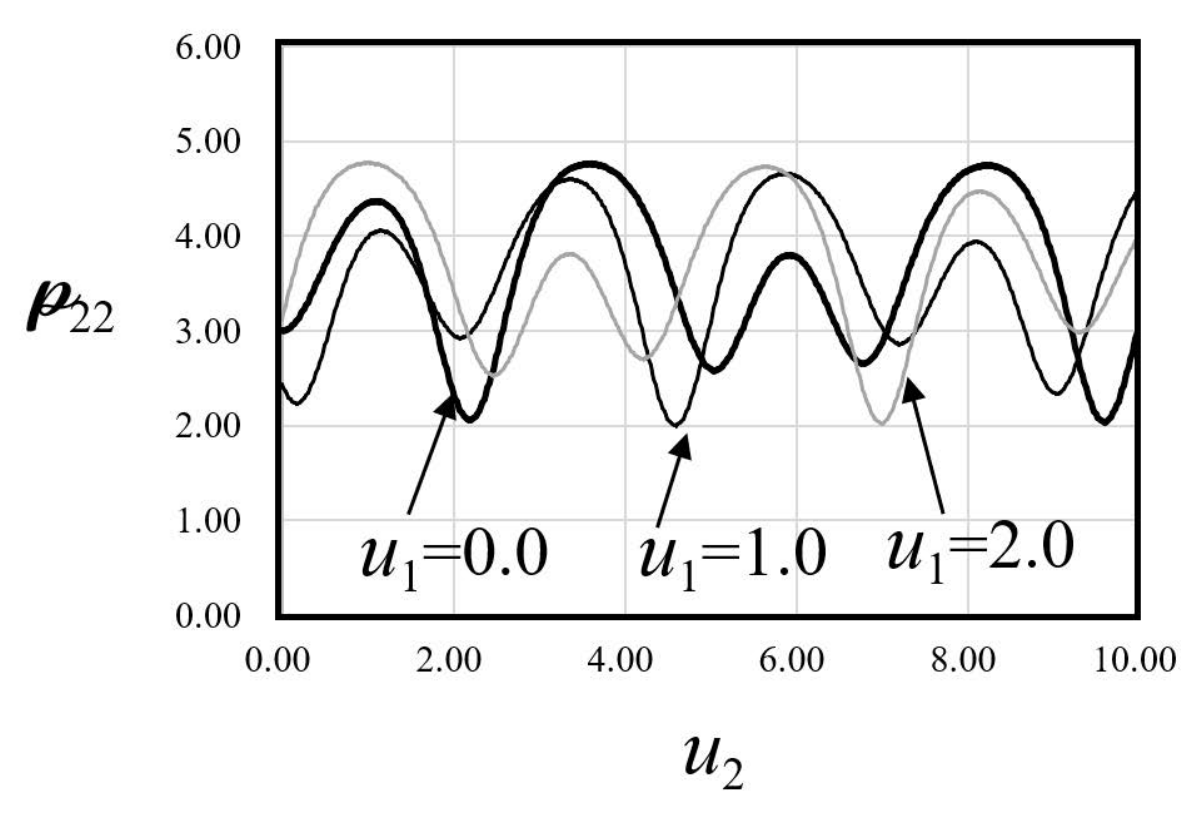}

(a) \hskip 0.45\hsize (b)

\end{center}

\caption{
The graphic representation of hyperelliptic functions of the curve
$(b_1, b_2, b_3, b_4, b_5) = (0, 1.8, 2, 3, 3.2)$:
We assume that $u_0:=w(B_1+B_4)$.
(a): $x_1(\begin{pmatrix}0 \\ u_2\end{pmatrix}+u_0)$ and 
$x_1(\begin{pmatrix}0 \\ u_2\end{pmatrix}+u_0)$.
(b): $\wp_{22}(\begin{pmatrix}u_1 \\ u_2\end{pmatrix}+u_0)$ for $u_1=0.0$,
$1.0$, and $2.0$.
}\label{fg:Fig02}
\end{figure}

\begin{figure}
\begin{center}

\includegraphics[width=0.38\hsize, bb=0 0 487 352]{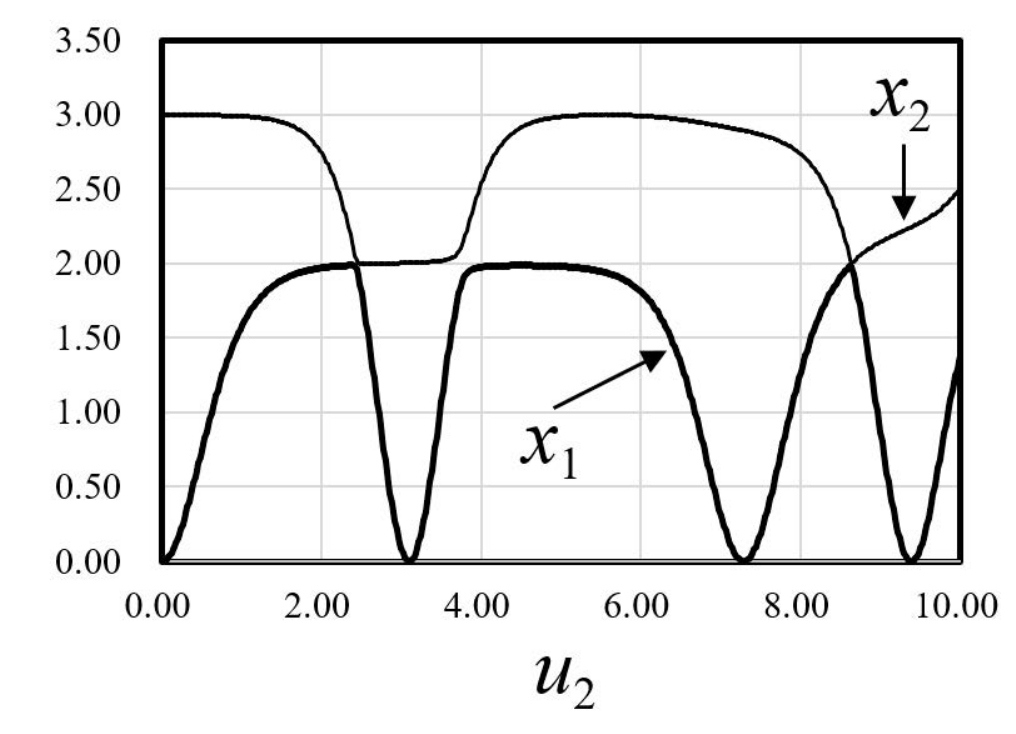}
\hskip 0.1\hsize
\includegraphics[width=0.41\hsize, bb=0 0 579 390]{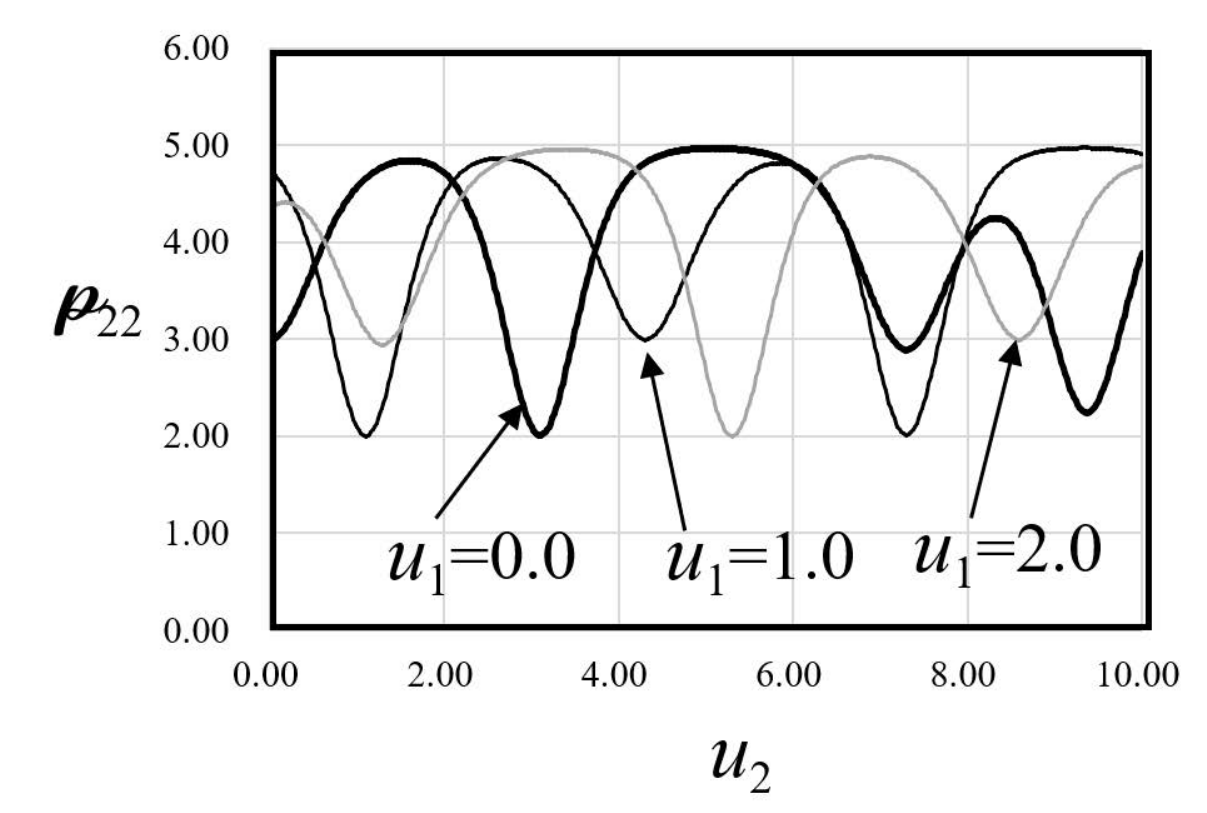}

(a) \hskip 0.45\hsize (b)

\end{center}

\caption{
The graphic representation of hyperelliptic functions of the curve
$(b_1, b_2, b_3, b_4, b_5) = and (0, 1.99, 2.0, 3, 3.01)$:
We assume that $u_0:=w(B_1+B_4)$.
(a): $x_1(\begin{pmatrix}0 \\ u_2\end{pmatrix}+u_0)$ and 
$x_1(\begin{pmatrix}0 \\ u_2\end{pmatrix}+u_0)$.
(b): $\wp_{22}(\begin{pmatrix}u_1 \\ u_2\end{pmatrix}+u_0)$ for $u_1=0.0$,
$1.0$, and $2.0$.
}\label{fg:Fig03}
\end{figure}

\section{Discussion and Conclusion}

In figures \ref{fg:Fig01}, \ref{fg:Fig02} and \ref{fg:Fig03} we have demonstrated the graphical representation of the hyperelliptic function $\wp_{22}$ of genus two for three cases.
Since we note that we are using only the data of the hyperelliptic curves as algebraic curves without any data of the theta functions, the computation is quite easy and its algorithm is simple.
There is no serious problem in choosing the computational parameters because the parameters in the computation are only the moduli parameters $b_1, \ldots, b_5$ and the initial points $(x_a, y_a)$, $(a=1,2)$: $y_a$ determines the directions of $\nuI{2}$.
(It is easy to check independence of the numerical parameter $\delta u_2$ due to good properties of the Abelian integrals.)

Since the computational method is very simple, it is easy to deal with these parameters of the curves and to find the parameter dependence.
Thus, we can numerically study the behavior of the degenerating limit of certain curves when they are real curves, as we have shown.
Furthermore, it is obvious that our method can be easily generalized to cases of higher genus and more complicated algebraic curves.
Though we have investigated the graphical representations of certain class of the quasi-periodic solutions of the modified KdV equation \cite{MP22, Mat23}, which contains crucial difficulties, it is easy to obtain the periodic solutions of the modified KdV equation for the real curves, i.e., $b_i \in \RR$.

In other words, it is clear that this method can be provided for various curves and allow the numerical investigations of periodic and quasi-periodic solutions of integrable equation.
We conclude that our graphical representation based on the Weierstrass-Baker theory of hyperelliptic functions is novel and useful.
It is expected that this method will be applied to various soliton equations and will lead to further developments in the algebraic study of integrable equations.

\bigskip
\bigskip

{\bf{Acknowledgment}:}
This research was born out of discussions with Saburo Kakei and Takashi Ichikawa at a workshop at Toyama Prefectural University, Oct. 12-14, 2023.
The recent results by Kakei \cite{Kakei} and the discussion made me realize the importance of applying this approach to real hyperelliptic curves.
I would like to express my sincere appreciation to both of them and to the organizers of the workshop.
I also acknowledge support from the Grant-in-Aid for Scientific Research (C) of Japan Society for the Promotion of Science, Grant No.21K03289.

\end{document}